\documentclass[prl,twocolumn,showpacs]{revtex4-1}
\usepackage[dvipdfmx]{graphicx}
\usepackage{bm}
\usepackage{amssymb}   
\usepackage{color}
\usepackage{amsmath}
\usepackage[colorlinks=true,linkcolor=blue]{hyperref}%
\usepackage{amsthm}
\newtheorem{theorem}{Theorem}
\newtheorem{lemma}[theorem]{Lemma}
\newtheorem{definition}[theorem]{Definition}

\begin{document}
\title{Canonical superdiffusion and energy fluctuation divergence}%
\author{Ken-ichi Okubo}%
\email{okubo.kenichi.65z@st.kyoto-u.ac.jp}
\author{Ken Umeno} %
\affiliation{Department of Applied Mathematics and Physics, 
Graduate School of Informatics, Kyoto University} %
\date{May 2016}%
\begin{abstract}
We propose a noble physical model obtained from a Hamiltonian with periodic potential.
This model is canonical, reversible and brings about chaotic superdiffusion with energy fluctuation divergence.
The analytical formula of invariant density can be obtained in some parameter range. 
In the range it is proved that the map is Anosov diffeomorphism and the invariant measure is a SRB measure.
We calculate the analytical formula of Lyapunov exponent.
\end{abstract}
\maketitle

\paragraph{Introduction}

Sinai-Ruelle-Bowen (SRB) measure which is a special case of Gibbs measure \cite{Young}
plays an important role in dynamical system and statistical dynamics points of view.
In the case of H\'enon map, Logistic map or Baker's map, it is proved that there is a SRB measure  
\cite{Benedicks,Jakobson,Tasaki} although in former two explicit form do not given and Baker's map does not have
time‐reversal symmetry.

Superdiffusion is also an important phenomenon in dynamical system and statistical dynamics.

In classical systems, there is no symplectic map in which it is proved analytically that superdiffusion occurs with measure unity. 
For example, in Standard map \cite{Venegeroles-Super}, it is said that superdiffusion occurs in some part of domain not in the whole domain.

In this paper, (i) new classical model whose SRB measure is given in explicit form in some parameter range
is proposed.
(ii) On the way to prove the existence of SRB measure, it is also proved that the map is
Anosov diffeomorphism and mixing in the parameter range. 
(iii) It is proved that energy fluctuation diverges because action variables are in accordance with Cauchy distribution. 
(iv) Lyapunov exponent is obtained using these properties.

At first Hamiltonian is introduced by
\begin{eqnarray}
& &H(I_1, I_2,\theta_1,\theta_2) \nonumber\\
&=& \frac{1}{2}\left(I_1^2 + I_2^2\right) -\varepsilon \log \left|\cos(\pi(\theta_1-\theta_2))\right|.
\end{eqnarray}

Let an interval $I$ be as $I = \left[-\frac{1}{2}, \frac{1}{2}\right)$
and the map $T$ is obtained by the Hamiltonian by leap frog method
\begin{eqnarray}
T &:& I^2 \times \mathbb{R}^2 \to  I^2 \times \mathbb{R}^2,\\
T\left(
\begin{array}{l}
I_1\\
\theta_1\\
I_2\\
\theta_2
\end{array}
\right)&=&
\left(
\begin{array}{l}
\displaystyle I_1 -\varepsilon \tan(\pi(I_1+\theta_1- I_2 -\theta_2)) \\
\displaystyle I_1+\theta_1  \mod{[-1/2, 1/2)}\\
\displaystyle I_2 + \varepsilon \tan(\pi(I_1+\theta_1-I_2-\theta_2))\\
\displaystyle  I_2+\theta_2  \pmod{[-1/2, 1/2)}
\end{array}
\right). \label{Symplectic map using Tan}
\end{eqnarray}

A two dimensional symplectic map with tangent function is also researched in \cite{Venegeroles}.
Figures  \ref{Fig: potential+} and \ref{Fig: potential-} show the behavior of the potential $V(\theta_1, \theta_2)$.
The absolute value of $V(\theta_1, \theta_2)$ diverges in $\{(\theta_1, \theta_2)| \cos(\pi(\theta_1-\theta_2))=0\}$.
\begin{figure}[h]
	\centering

	\includegraphics[width=8cm]{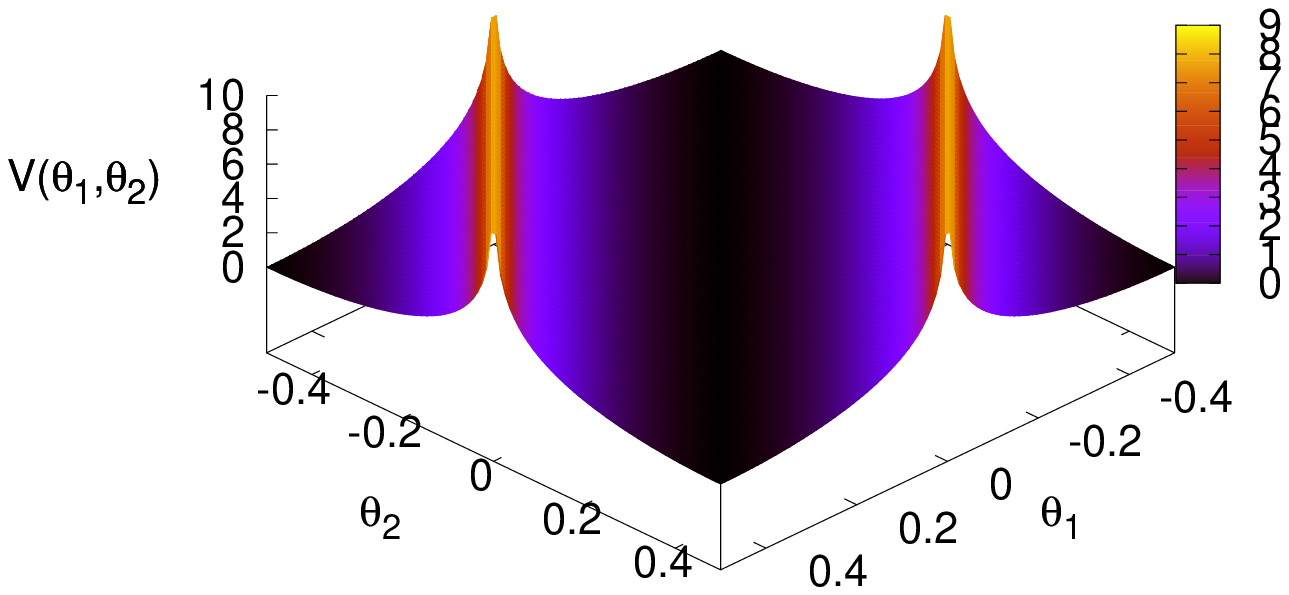}

	\caption{The shape of potential $V(\theta_1, \theta_2) = -\varepsilon\log|\cos(\pi(\theta_1-\theta_2))|$ for $\varepsilon=2.0$.
		The potential diverges in $\{(\theta_1, \theta_2)| \cos(\pi(\theta_1-\theta_2))=0\}$.}
	\label{Fig: potential+}

	\includegraphics[width=8cm]{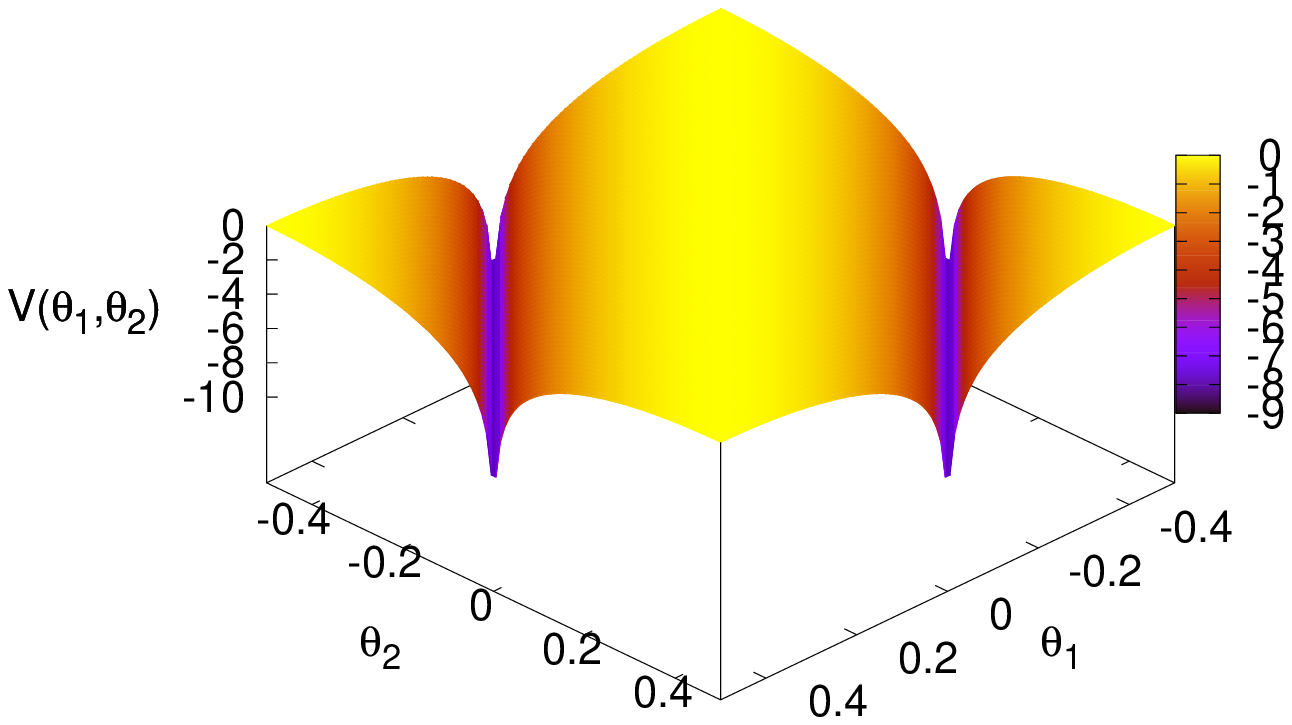}

	\caption{The shape of potential $V(\theta_1, \theta_2) = -\varepsilon\log|\cos(\pi(\theta_1-\theta_2))|$ for $\varepsilon=-2.0$.}
	\label{Fig: potential-}
	
\end{figure}
This map has a conserved quantity for momentum $C$ such that
\begin{eqnarray*}
C &\equiv&  I_1(0) + I_2(0)= \cdots
= I_1(n) + I_2(n) = \cdots.
\end{eqnarray*}
We can reduce dgree of freedom by using $C$ as
\begin{eqnarray}
I_2(n) &=& C - I_1(n),\label{I2}\\
\theta_1(n) + \theta_2(n) &=&  \theta_1(0) + \theta_2(0) + nC,\nonumber 
\end{eqnarray}
Then, we get such formulus as Equations (\ref{Formula1}) and (\ref{Formula2}).
\begin{figure*}
\begin{eqnarray}
\theta_1(n+1) - 2\theta_1(n) + \theta_1(n-1)  &=& -\varepsilon\tan\left( \pi(\theta_1(n) -\theta_2(n))\right), \label{Formula1}\\
\theta_2(n+1) - 2\theta_2(n) + \theta_2(n-1)  &=& \varepsilon\tan\left( \pi(\theta_1(n) -\theta_2(n))\right). \label{Formula2}
\end{eqnarray}
\end{figure*}
By subtracting Eq. (\ref{Formula2}) from Eq. (\ref{Formula1}), we obatain Eq. (\ref{Formula3})
\begin{figure*}
	\begin{eqnarray}
	\left[ \theta_1(n+1) -\theta_2(n+1)\right] - 2\left[ \theta_1(n)-\theta_2(n)\right] + \left[ \theta_1(n-1)-\theta_2(n-1)\right]  = -2\varepsilon\tan\left[ \pi(\theta_1(n) -\theta_2(n))\right], \label{Formula3}
	\end{eqnarray}
\end{figure*}
Now by changing variables as
\begin{eqnarray}
\begin{array}{cccc}
p_n &\equiv& \theta_1(n-1)-\theta_2(n-1) &\mod{[-1/2, 1/2)},\\
q_n &\equiv& \theta_1(n) -\theta_2(n) &\mod{[-1/2, 1/2)},
\end{array}
\end{eqnarray}
we get another equation which is topologically conjugate with  Eq. (\ref{Formula1}) as 
\begin{eqnarray}
\tilde{T} : x_n \longmapsto \tilde{T}x_n, x_n \in I \times I,
\end{eqnarray}
\begin{eqnarray}
	\begin{split}
&\left( 
\begin{array}{c}
p_{n+1}\\
q_{n+1}
\end{array}
\right) =
\tilde{T}\left( 
\begin{array}{c}
p_n\\
q_n
\end{array}
\right) \\=
&\left( 
\begin{array}{c}
q_n \mod{[-1/2, 1/2)}\\
2q_n -p_n -2 \varepsilon \tan( \pi q_n) \mod{[-1/2, 1/2)}
\end{array}\right) 
\end{split} \label{Transform formula}
\end{eqnarray}
The Jacobian  of Eq. (\ref{Transform formula}) is given as
\begin{eqnarray}
J(n) = \left(
\begin{array}{cc}
0 & 1\\
-1 & 2- \frac{2\pi \varepsilon}{\cos^2(\pi q_n)}
\end{array}
\right) \label{Jacobi}
\end{eqnarray}
The local instability condition (for at least one eigenvalue of the matrix \ref{Jacobi},
its absolute value is larger than unity)
 for (\ref{Jacobi}) is as
\begin{eqnarray}
\varepsilon<0, \frac{2}{\pi}< \varepsilon. \label{Condition}
\end{eqnarray}

\begin{theorem}
	The uniform distribution 
	\begin{eqnarray}
	\rho(p, q) = 1 \label{uniform distribution}
	\end{eqnarray}
	is an invariant density of the map $T_\varepsilon$ on the manifold $I\times I$.
\end{theorem}

\begin{proof}
	We prove the proposition by showing
	the uniform distribution $\rho$ is a solution of 
	 Perron-Frobenius Equation 
	\begin{eqnarray}
	\begin{split}
	f(x, y) &= \int \int_{-1/2}^{1/2} f(p, q) 
	\delta\left(
	X - T_\varepsilon(Y)
	\right)
	dpdq,\\
	\mbox{where}~~  X &= (x, y),\\
	 Y &= (p, q).
	\end{split}	
	\end{eqnarray}
	
	When a uniform distribution 
	satisfies Perron-Frobenius Equation, uniform distribution (\ref{uniform distribution}) is an invariant density for 
	$(I\times I, \tilde{T}_\varepsilon)$. Because the values of $\rho(x, y)$ and $\rho(p, q)$ is equivalent,
	we show for any point $X=(x, y)$ there is only one point $(p', q')$ which satisfies $X- T_\varepsilon(p', q') = \mathbf{0}$.

	First, $X$ is fixed. $x$ is on the $I = [-1/2, 1/2)$. Then we can determine only one $q \in I$ which satisfies
	$x-q=0$. Then $x, y$ and $q$ are fixed. Then there is only one $p \in I$ to satisfies 
	\begin{eqnarray}
	y - \left[2q -p -2 \varepsilon\tan(\pi q)  \mod{[-1/2, 1/2)}\right]=0.
	\end{eqnarray}
	Therefore uniform distribution is a solution of the Perron-Frobenius Equation.
\end{proof}

When $p$ and $q$ distribute uniformly, $\varepsilon \tan(\pi p)$ or $\varepsilon \tan(\pi q)$
are in accordance with the Cauchy distribution as
\begin{eqnarray}
f(x) = \frac{1}{\pi} \frac{|\varepsilon|}{x^2 + |\varepsilon|^2}.
\end{eqnarray}

\begin{theorem}
	The probability variables $p$ and $q$ are independent.
\end{theorem}
\begin{proof}
	For any $L^1$ class function $A$ and $B$,
	\begin{eqnarray*}
	\mathbb{E}[A(p)B(q)] &=& \int_p\int_q A(p)B(q) \left(\frac{1}{4}dpdq\right),\\
	&=& \int_p A(p) \left(\frac{1}{2}dp\right) \int_q B(q) \left(\frac{1}{2}dq\right),\\
	&=& \mathbb{E}[A(p)] \mathbb{E}[B(q)].
	\end{eqnarray*}
Then, $p$ and $q$ are independent.
\end{proof}

\paragraph{Mixing property}

Let consider a set $A$ defined by
\begin{eqnarray*}
A = \left\lbrace (p, q) | {}^\exists n \in \mathbb{Z}~ \mbox{s.t.}~  q' = \pm 1/2, (p', q') = \tilde{T}_\varepsilon^n(p, q)\right\rbrace.
\end{eqnarray*}
Then, define manifold $M = (I\times I) \backslash A$.

\begin{definition}
	A diffeomorphism $f : M \to M$ where $M$ is a closed manifold is Anosov diffeomorphism when there exists a direct
	sum decomposition of the tangent bundle $T_xM$ at each point $x$ into complementary subspace $E_x^u, E_x^s$ such that
	\begin{eqnarray}
	(D_xf)E_x^u &=& E_{f(x)}^u,\\
	(D_xf)E_x^s &=& E_{f(x)}^s,\\
	\bm{\xi}\in E_x^u \Longrightarrow\|(D_xf^n)E_x^u\| &\geq& K \lambda^n \|\bm{\xi}\|, \label{stretching}\\ 
	\bm{\xi}\in E_x^s \Longrightarrow\|(D_xf^n)E_x^s\| &\leq&  K \lambda^n \|\bm{\xi}\|, \label{shrinking}
	\end{eqnarray}
	where $K>0, 0<\lambda<1$ are determined by $x$ not by $\bm{\xi}$ or $n$.
\end{definition}

\begin{lemma}
	When the condition (\ref{Condition}) is satisfied,
	the map $\tilde{T}_\varepsilon$ on $M$ is an Anosov diffeomorphism.
\end{lemma}

According to \cite{Zaslavsky,Arnold}, for any normalized two dimensional vector 
\begin{eqnarray*}
\mathbf{a}= (a_1(n), a_2(n)),
\end{eqnarray*}
consider cones $L^+$ and $L^-$ such that
\begin{eqnarray*}
L^+  &=& \{(a_1(n), a_2(n)); \| J(n)\bm{a}\| > \|\bm{a}\|\},\\
L^-  &=& \{(a_1(n), a_2(n)); \|J(n)\bm{a}\| < \|\bm{a}\|\}.
\end{eqnarray*}

The goal is to prove
\begin{eqnarray}
J(n)\bm{a} \in L^+(\tilde{T}_\varepsilon x_n), {}^\forall\bm{a} \in L^+(x_n), \label{Mixing condition1}\\
  J(n-1)^{-1}\bm{a} \in L^-(\tilde{T}_\varepsilon^{-1}x_n), {}^\forall \bm{a} \in L^-(x_n). \label{Mixing condition2}
\end{eqnarray}

\begin{proof}
	In the case of ($\ref{Mixing condition1}$), 
	the condition that $\bm{a} \in L^+(x_n)$ is expressed by
	\begin{eqnarray}
	\bm{a} \in L^+(x_n)
	\Leftrightarrow \left\lbrace 
	\begin{array}{cc}
	\frac{a_1(n)}{a_2(n)} > 1-\frac{\pi\varepsilon}{\cos^2(\pi q_n)},& \varepsilon > \frac{2}{\pi}\\
	\frac{a_1(n)}{a_2(n)} < 1-\frac{\pi\varepsilon}{\cos^2(\pi q_n)},& \varepsilon < 0.
	\end{array}
	\right. \label{L^+}
	\end{eqnarray}
	Let define $\bm{a}' \equiv J(n)\bm{a} = (a_1(n+1), a_2(n+1))$. The condition that $\bm{a}' \in L^+(x_{n+1})$ is expressed by
	\begin{eqnarray}
	& &J(n)\bm{a} \in L^+(\tilde{T}_\varepsilon x_n),\nonumber\\
	&\Leftrightarrow& \left\lbrace
	\begin{array}{c}
	\frac{a_1(n+1)}{a_2(n+1)} > 1-\frac{\pi\varepsilon}{\cos^2(\pi q_{n+1})},~\varepsilon > \frac{2}{\pi}\\
	\frac{a_1(n+1)}{a_2(n+1)} < 1-\frac{\pi\varepsilon}{\cos^2(\pi q_{n+1})},~\varepsilon < 0.
	\end{array}
	\right. \label{L^+2}
	\end{eqnarray}

	By substituting $a_1(n+1)=a_2(n)$ and $a_2(n+1)= -a_1(n) +2\left(1-\frac{\pi\varepsilon}{\cos^2(\pi q_n)}\right)a_2(n)$ and
	considering
	\begin{eqnarray*}
	\frac{a_1(n+1)}{a_2(n+1)}&=&\frac{a_2(n)}{-a_1(n) +2\left(1-\frac{\pi\varepsilon}{\cos^2(\pi q_n)}\right)a_2(n)},\\
	&=& \frac{1}{-\frac{a_1(n)}{a_2(n)} +2\left(1-\frac{\pi\varepsilon}{\cos^2(\pi q_n)}\right)}
	\end{eqnarray*}
	Then, when the condition (\ref{L^+}) is satisfied, it holds that
	\begin{eqnarray}
	1-\frac{\pi\varepsilon}{\cos^2(\pi q_{k})}<\frac{a_1(k)}{a_2(k)}<0,~{}^\forall k\geq n+1, \varepsilon>\frac{2}{\pi},\\
	0<\frac{a_1(k)}{a_2(k)}<1-\frac{\pi\varepsilon}{\cos^2(\pi q_{k})},~{}^\forall k\geq n+1, \varepsilon<0
	\end{eqnarray}
	Therefore the condition (\ref{Mixing condition1}) holds. Then consider a subset $LL^+(x_n)$ of $L^+(x_n)$ defined by
	\begin{eqnarray*}
	LL^+(x_n) 	=
	 \left\lbrace (a_1(n), a_2(n)) \Bigg|  1-\frac{\pi\varepsilon}{\cos^2(\pi q_{n})}<\frac{a_1(n)}{a_2(n)}<0\right\rbrace 
	\end{eqnarray*}
	when $\varepsilon> \frac{2}{\pi}$ and
	\begin{eqnarray*}
	LL^+(x_n) 	=
	\left\lbrace (a_1(n), a_2(n)) \Bigg|  0<\frac{a_1(n)}{a_2(n)}<1-\frac{\pi\varepsilon}{\cos^2(\pi q_{n})}\right\rbrace 
	\end{eqnarray*}
	when $\varepsilon<0$. It holds that
	\begin{eqnarray}
	J(n)\bm{a} \in LL^+(\tilde{T}_\varepsilon x_n), {}^\forall\bm{a} \in LL^+(x_n).
	\end{eqnarray}
	
	In the case of (\ref{Mixing condition2}), it holds that
	\begin{eqnarray}
	\bm{a} \in L^-(x_n) \Leftrightarrow
	\left\lbrace 
	\begin{array}{cc}
		\frac{a_1}{a_2} > 1-\frac{\pi\varepsilon}{\cos^2(\pi q_n)},& \varepsilon > \frac{2}{\pi},\\
		\frac{a_1}{a_2} < 1-\frac{\pi\varepsilon}{\cos^2(\pi q_n)},& \varepsilon < 0.
	\end{array}
	\right. \label{L^-}
	\end{eqnarray}
	Then, 
	\begin{eqnarray}
	& &J(n-1)^{-1}\bm{a} \in L^-(\tilde{T}_\varepsilon^{-1}x_n),\nonumber\\
	&\Leftrightarrow& \left\lbrace 
	\begin{array}{cc}
	\frac{a_1(n-1)}{a_2(n-1)} > 1-\frac{\pi\varepsilon}{\cos^2(\pi q_{n-1})},& \varepsilon > \frac{2}{\pi}\\
	\frac{a_1(n-1)}{a_2(n-1)} < 1-\frac{\pi\varepsilon}{\cos^2(\pi q_{n-1})},& \varepsilon < 0.
	\end{array}
	\right. \label{L^-2}
	\end{eqnarray}
	By substituting $a_1(n-1) = 2\left(1-\frac{\pi\varepsilon}{\cos^2(\pi q_{n-1})}\right)a_1(n)-a_2(n)$ and 
	$a_2(n-1) = a_1(n)$, one can see condition (\ref{L^-2}) holds when condition (\ref{L^-}) holds.
	
	Then, by considering a orbit $\{x_n\}_{n=-\infty}^\infty$ there exists a cone $L^-(x_\infty)$ and $l^-(x_n)$ such that
	\begin{eqnarray*}
		L^-(x_{n}) &\supset& J^{-1}(n)L^-(x_{n+1}),\\
	&\supset&  J^{-1}(n)J^{-1}(n+1)L^-(x_{n+2}),\\
	&\cdots& \\
	&\supset& l^-(x_n),\\
	l^-(x_n) &\equiv&  \left(\prod_{k=n}^{\infty}J^{-1}(k)\right)L^-(x_\infty),
	\end{eqnarray*}
	
	Then, one can choose any eigenvector spaces $E_x^u$ and $E_x^s$ from $LL^+(x_n)$ and $l^-(x)$ each other by
	\begin{eqnarray}
	E_{x_n}^u \subset LL^+(x_n),	E_{x_n}^s \subset l^-(x_n).
	\end{eqnarray}
	 It is established that
	$\left( D_{x_n}\tilde{T}_\varepsilon\right)  E_{x_n}^u \subset LL^+(x_{n+1})$ and 
	$\left( D_{x_n}\tilde{T}_\varepsilon\right)  E_{x_n}^s \subset L^-(x_{n+1}).$ 
	Then, one can define $E_{x_{n+1}}^u$, $E_{x_{n+1}}^s$  by
	\begin{eqnarray}
	E_{x_{n+1}}^u \equiv \left( D_{x_n}\tilde{T}_\varepsilon\right)  E_{x_n}^u, 
	E_{x_{n+1}}^s \equiv \left( D_{x_n}\tilde{T}_\varepsilon\right)  E_{x_n}^s.
	\end{eqnarray}
	$E_{x_n}^u$ and $E_{x_n}^s$ are independent by theire definition so that, 
	it holds that
	\begin{eqnarray*}
	T_{x_n}M = E_{x_n}^u\oplus E_{x_n}^s.
	\end{eqnarray*}
	Next let's determine $K, \lambda$. Stretching rate $\sigma$ is defined by
	{\small
	\begin{eqnarray*}
	& &\sigma(x_n, \bm{a}) = \frac{\|J(n)\bm{a}\|^2}{\|\bm{a}\|^2},\\
	&=&\frac{a_1^2+a_2^2 -4a_1a_2\left(1-\frac{\pi\varepsilon}{\cos^2(\pi q_n)}\right)+4\left(1-\frac{\pi\varepsilon}{\cos^2(\pi q_n)}\right)^2a_2^2}{a_1^2+a_2^2},\\
	&=& 1 -4\frac{a_1a_2}{a_1^2+a_2^2}\left(1-\frac{\pi\varepsilon}{\cos^2(\pi q_n)}\right)+4\left(1-\frac{\pi\varepsilon}{\cos^2(\pi q_n)}\right)^2\frac{a_2^2}{a_1^2+a_2^2},\\
	&=& 1+4\left(1-\frac{\pi\varepsilon}{\cos^2(\pi q_n)}\right)\left\lbrace \left(1-\frac{\pi\varepsilon}{\cos^2(\pi q_n)}\right)\sin^2\phi-\sin\phi\cos\phi\right\rbrace ,
	\end{eqnarray*}
}
	where $\sin^2\phi \equiv \frac{a_2^2}{a_1^2+a_2^2}, \sin\phi\cos\phi \equiv \frac{a_1a_2}{a_1^2+a_2^2},~-\pi< \phi\leq\pi$. 
	Then, it holds that if $\bm{a}\in LL^+(x_n)$,
	{\small
	\begin{eqnarray}
	\left(1-\frac{\pi\varepsilon}{\cos^2(\pi q_n)}\right)\sin^2\phi-\sin\phi\cos\phi
	\left\lbrace 
	\begin{array}{cc}
	<0, &\varepsilon>\frac{2}{\pi},\\
	>0, &\varepsilon<0,
	\end{array}\right.
	\end{eqnarray}
}
	Let define $\alpha_n \equiv \left(1-\frac{\pi\varepsilon}{\cos^2(\pi q_n)}\right)$, and
	\begin{eqnarray*}
	g(\phi) &\equiv& \alpha_n \sin^2\phi-\sin\phi\cos\phi,\\
	g'(\phi) &=& \alpha_n \sin(2\phi) -\cos(2\phi),\\
	&=& \sin(2\phi)\left(\alpha_n - \cot(2\phi)\right).
	\end{eqnarray*}
	
	(I) Case of $\varepsilon> \frac{2}{\pi}$,
	
	Considering 
	$\alpha_n<\frac{a_1(n)}{a_2(n)}= \cot\phi_n<0$,
	the range of $\phi_n$ is expressed by
	\begin{eqnarray}
	\frac{\pi}{2} < \phi_n < \psi_n, -\frac{\pi}{2} < \phi_n < \psi_n-\pi,
	\end{eqnarray}
	where $\cot\psi_n = \alpha_n$.  Since $g'(\phi)$ is positive in this range,
	It becomes 
	\begin{eqnarray*}
	g(\phi_n) &<& g(\psi_n) = \frac{2\alpha_n}{1+\alpha_n^2}<0,\\
	\sigma(x_n, \bm{a}) &<& 1 +4\alpha_n \cdot \frac{2\alpha_n}{1+\alpha_n^2}
	<1+ \frac{16(1-\pi\varepsilon)^2}{1+(1-\pi\varepsilon)^2}
	\end{eqnarray*}
	Then by defining $K = 1$ and $\lambda = \sqrt{1+ \frac{16(1-\pi\varepsilon)^2}{1+(1-\pi\varepsilon)^2}}$,
	condition (\ref{stretching}) is satisfied.

	(II) Case of $\varepsilon <0$,
	
	Considering $0< \cot\phi_n < \alpha_n$,
	the range of $\phi_n$ is expressed by
	\begin{eqnarray}
	\psi_n < \phi_n < \frac{\pi}{2}, \psi_n -\pi < \phi_n < -\frac{\pi}{2},
	\end{eqnarray}
	Since $g'(\phi)$ is positive in this range,
	It becomes 
	\begin{eqnarray*}
		g(\phi_n) &>& g(\psi_n) = \frac{2\alpha_n}{1+\alpha_n^2}>0,\\
		\sigma(x_n, \bm{a}) &>& 1 +4\alpha_n \cdot \frac{2\alpha_n}{1+\alpha_n^2}
		>1+ \frac{16(1-\pi\varepsilon)^2}{1+(1-\pi\varepsilon)^2}
	\end{eqnarray*}
	Then by defining $K = 1$ and $\lambda = \sqrt{1+ \frac{16(1-\pi\varepsilon)^2}{1+(1-\pi\varepsilon)^2}}$,
	condition (\ref{stretching}) is satisfied.
	
	For a symplectic map, since the shrinking rate is a inverse of stretching rate,
	condition (\ref{shrinking}) is also holds.
	Therefore, the map $\tilde{T}_\varepsilon$ on $M$ is an Anosov diffeomorphism.
\end{proof}

	According to \cite{Arnold}, Anosov diffeomorphism is K-system. Therefore the theorem below holds.
\begin{theorem}
	When the condition (\ref{Condition}) is satisfied,
	dynamical system $(M, \tilde{T}_\varepsilon, \frac{1}{4}dpdq)$ has the mixing property.
\end{theorem}

According to \cite{Young}, if $(M, \tilde{T}_\varepsilon)$ is a Anosov diffeomorphism,
$M$ is also a Axiom A attractor. 
Since Lebesgue measure $dpdq$ is preserved by $\tilde{T}_\varepsilon$ on $M$, and 
dynamical system $(M, \tilde{T}_\varepsilon, dpdq)$ is ergodic for $\varepsilon<0, \frac{2}{\pi}< \varepsilon$.
Therefore, Lebesgue measure is a unique SRB measure.

\paragraph{Cauchy distribution}
Action variable $I_1(n)$ can be expressed by
\begin{eqnarray}
I_1(n) = I_1(0) -\varepsilon \sum_{k=1}^{n}\tan \left[\pi\left\lbrace \theta_1(k)-\theta_2(k)\right\rbrace \right].
\end{eqnarray}
A probability variable 
$s_k=\{-\varepsilon\tan\left[\pi\left\lbrace \theta_1(k)-\theta_2(k)\right\rbrace \right]\}$
is according to the Cauchy distribution whose scale parameter is $|\varepsilon|$ 
for $\varepsilon<0, \frac{2}{\pi}<\varepsilon$. That is 
$\{s_k\}$ is stationary and strongly mixing. Then according to \cite{Ibragimov}, 
$\{I_1(n)\}$ is in accordance with a stable distribution.
Figure \ref{Fig: Cauchy Action} shows log-log plot of the distribution $f(x)$ of $\{I_1(100)\}$ at $\varepsilon = 0.65$
	obtained by numerical experiment and fitted function $g(x)$.
	The number of initial points $N$ is $N=10^6$. 
	When $\{I_1(100)\}$ are in accordance with Cauchy distribution whose scale parameter is $a$, 
	the probability variables $\{x=(I_1(100)-\mu)/\sqrt{\sigma}\}$ are considered to be in accordance with $g(x)$ defined by
	\begin{eqnarray}
	g(x) = \frac{1}{\pi} \frac{a \sqrt{\sigma}}{(\sqrt{\sigma}x)^2 + a^2},
	\end{eqnarray}	
	where $\mu$ and $\sqrt{\sigma}$ are an average and a variance respectively obtained from finite number of probability variables $\{I_1(100)\}$. 
	By fitting $g(x)$ to the data using least squares method, fitted parameter $\hat{a}$ is obtained as
	$\hat{a} = 68.01 \simeq 0.65 \times 100$. This result shows $I_1$ is in accordance with Cauchy distribution, so that
	the true average and variance of $I_1$ do not exist.

\begin{figure}
	\centering
	\includegraphics[width=.9\columnwidth]{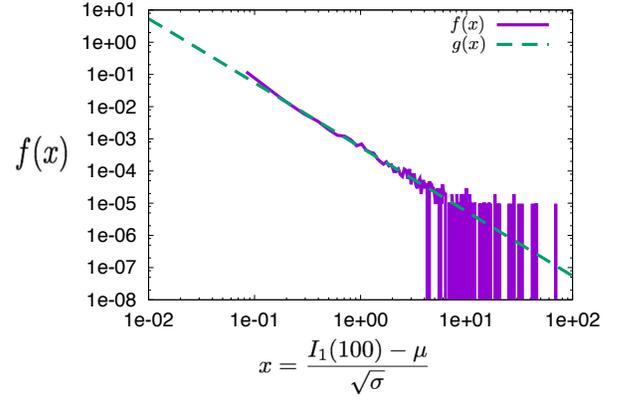}
	
	\caption{The density function $f(x)$ of $x=I_1(100)$. Initial points $\{(p_0, q_0)\}$ are uniformly distributed and
		$I_1(100)$ is obtained by $I_1(100) = -\varepsilon \sum_{i=0}^{99}\tan(\pi q_i)$. 
		$g(x)$ and parameter $\hat{a}$ obtained by least squares method is $\hat{a} = 68.01 \simeq 100\times 0.65$.
		$\mu = 2.64, \sqrt{\sigma}= 3.93\times 10^4$}
	
	\label{Fig: Cauchy Action}
\end{figure} 

\paragraph{The fluctuation of Energy}
In not integrable system, there is a trade-off relation between the conservation of Energy and the that of symplecticity \cite{Ge}.
Then, this symplectic map, the energy cannot be conserved and fluctuates. Especially in this map since the distribution with 
$I_1$ and $I_2$ are in accordance with Cauchy distribution, their variances $\sigma(I_{1,2})$ diverge. Then the kinetic energy and a total one
also diverge. The Figure \ref{Fig: Fluctuation} shows the behavior of the fluctuation of energy.

\begin{figure}[h]
	\centering
	\hspace*{1cm}
	\includegraphics[width=8cm]{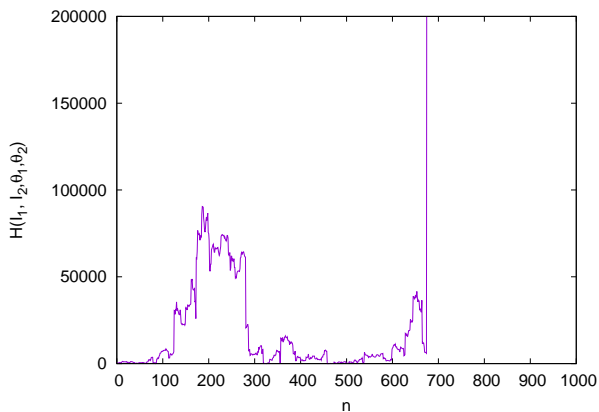}
	
	\vspace*{1cm}
	\caption{The time behavior of total Energy $H(I_1, I_2, \theta_1, \theta_2)$. It fluctuates and its variance diverges.
		The initial condition is $(I_1, I_2, \theta_1, \theta_2)=(1.41, 1.51+\frac{\pi}{2}, 0.1, -1.4)$.}
	\label{Fig: Fluctuation}
\end{figure}

\paragraph{Superdiffusion}
 Considering that $I_1$ is in accordance with Stable distribution and referring the Figure \ref{Fig: Cauchy Action},
 superdiffusion occurs. 
Figure \ref{Fig: MSD} shows the log-log plot with times evolution of 
Mean Square Displacement (MSD) for $I_1$ at $\varepsilon = 0.66> \frac{2}{\pi}$ 
and there occurs superdiffusion.
The inclination of the fitting line is about $1.80>1$.

\begin{figure}[h]
	\centering
	
\hspace*{1cm}
	\includegraphics[width = 8cm]{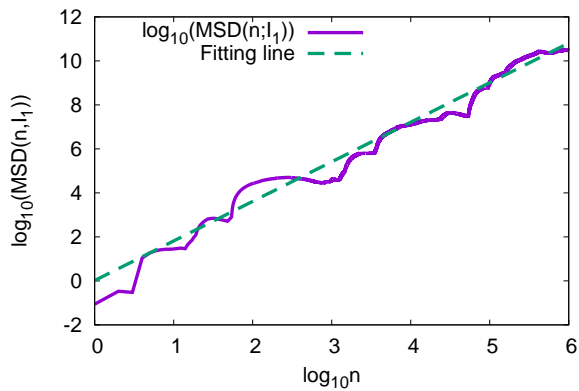}
	
	\vspace*{1cm}
	\caption{The mean square displacement (MSD) of map $T$ for $\varepsilon=0.66$  (solid line) and its fitting line (broken line). 
		The inclination of fitting line is about $1.80$.
		Super diffusion occurs.}
	\label{Fig: MSD}
\end{figure}

\paragraph{Lyapunov exponent}
Since the Lebesgue measure is unique SRB mesure for $\varepsilon<0, \frac{2}{\pi}<\varepsilon$,
KS entropy $h(\tilde{T}_\varepsilon)$ is expressed as
\begin{eqnarray}
h(\tilde{T}_\varepsilon) &=& \int \int \log |\det \left( D\tilde{T}_\varepsilon|_{E^u}\right) |dpdq,\\
&=& \int \int \log |\gamma |dpdq,
\end{eqnarray}
where $\gamma$ is a eigenvalue of Jacobian $J(x)$ whose absolute value is larger than unity.
\begin{eqnarray}
\gamma &=& \left\lbrace 
\begin{array}{c}
\gamma_-,~\mbox{when}~\varepsilon > \frac{2}{\pi},\\
\gamma_+,~\mbox{when}~\varepsilon < 0,
\end{array}
\right.\\
\gamma_{\pm} &=& 1- \frac{\pi\varepsilon}{\cos^2(\pi q)}\pm \sqrt{\left(1-\frac{\pi\varepsilon}{\cos^2(\pi q)}\right)^2-1}.
\end{eqnarray}
$h(\tilde{T}_\varepsilon)$ is equivalent to a positive Lyapunov exponent.
Then, in $|\varepsilon| \gg 1$, the Lyapunov exponent can be expressed by
\begin{eqnarray}
\lambda(\varepsilon) = \log \left|2\left(1-2\varepsilon \pm \sqrt{4\varepsilon(\varepsilon-1)}\right)\right|.
\end{eqnarray}
Figure \ref{Fig: Lyapunov exponent} shows the comparison between the numerical result and analytical formula of Lyapunov exponent.
The numerical result is consistent with analytic formula in $|\varepsilon| \gg 1.$
\begin{figure}[!h]
	\centering
	\vspace*{-1cm}
\hspace*{1cm}
	\includegraphics[width = 8cm]{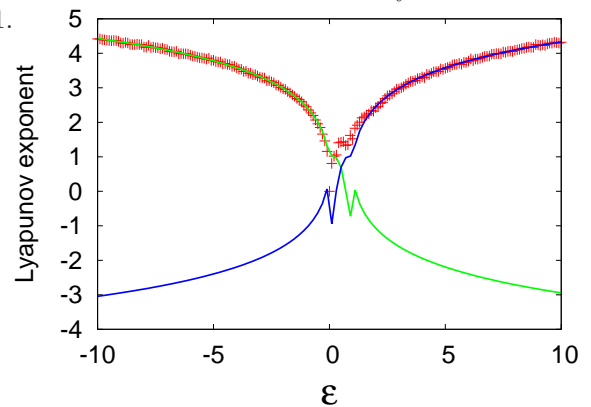}
	
	\vspace*{1cm}
	\caption{The numerical result (red cross) of Lyapunov exponent and analytical formula (solid line).}
	\label{Fig: Lyapunov exponent}
\end{figure}
\paragraph{Conclusion}
We proposed canonical deterministic model with superdiffusion. 
We showed the condition of superdiffusion. The analytical formula of invariant density is obtained and mixing property
is shown when this condition satisfied. We also calculate Lyapunov exponent and it is consistent with a numerical 
experiment. The action variables are in accordance with the Cauchy distribution whose scale parameter is $|\varepsilon|$ in 
$\varepsilon<0, \frac{2}{\pi}< \varepsilon$. Therefore, this model has energy fluctuation divergence.


\begin{thebibliography}{99}
 \bibitem{Young} L. S. Young, \textit{Journal of Statistical Physics} \textbf{108}, 733-754 (2002).
 \bibitem{Benedicks} M. Benedicks and L. Carleson, \textit{Annals of Mathematics Second Series} \textbf{133}, 73--169 (1991).
 \bibitem{Jakobson} M. V. Jakobson, Construction of invariant measures absolutely continuous with respect to dx for some maps of the interval,
  \textit{Global Theory of Dynamical Systems} Springer Berlin Heidelberg, 246-257 (1980).
  \bibitem{Tasaki} S. Tasaki, T. Gilbert and J. R. Dorfman, \textit{Chaos An Internationary Journal of Nonlinear Science} \textbf{8}, 424--443 (1998).
  \bibitem{Venegeroles-Super} R. Venegeroles, \textit{Physical Review Letters} \textbf{101}, 054012 (2008).
 	\bibitem{Venegeroles}R. Venegeroles and A. Saa, \textit{Journal of Statistical Mechanics: Theory and Experiment} \textbf{2008 (01)}, P01005 (2008).  
 \bibitem{Zaslavsky}G. M. Zaslavsky, \textit{Chaos in Dynamical Systems}, Harwood Academic Publishers, 1984.
 \bibitem{Arnold}
 V. I. Arnold and A. Avez, \textit{Ergodic problems of classical mechanics}, Benjamin, 1968. 
 \bibitem{Ricardo} R. Ma\~n\'e, \textit{Ergodic Theory and Differentiable Dynamics}, Springer-Verlag, 1983. 
 \bibitem{Ge} G. Zhong and J. E. Marsden, \textit{Physics Letters A} \textbf{133}, 134--139 (1988).
 \bibitem{Ibragimov} I. A. Ibragimov, \textit{Theory of Probability and Its Applications} \textbf{7}, 349–382 (1962).
\bibitem{Okubo16} K. Umeno and K. Okubo, \textit{Progress of Theoretical and Experimental Physics} (Letter) 
\textbf{021A01}, (2016).
\end{thebibliography}
\end{document}